\documentclass[showpacs,pra,tightenlines,12pt]{revtex4}
\usepackage{graphics}
\usepackage{amsmath}
\usepackage{amssymb}
\usepackage{amsthm}
\usepackage{mathrsfs}
\usepackage{amsfonts}
\usepackage{overpic}
\usepackage{hyperref}
\usepackage{rotating}
\usepackage[english]{babel}
\usepackage{natbib}

\newtheorem{thm}{Theorem}[section]

\newtheorem{lem}[thm]{Lemma}

\newtheorem{prop}[thm]{Property}
\newtheorem{rem}[thm]{Remark}
\newtheorem{num}[thm]{Numerical Result}
\newtheorem{numprop}[thm]{Numerical Result}

\newcommand{\abs}[1]{\left\vert#1\right\vert}

\newcommand{\Real}{\mathbb R}
\newcommand{\eps}{\varepsilon}

\newcommand{\qubit}[1]{\ensuremath{{|#1\rangle}}}
\newcommand{\qubitc}[1]{\ensuremath{{\langle#1|}}}
\newcommand{\qrho}[1]{\qubit{#1}\qubitc{#1}}
\newcommand{\qpsi}{\ensuremath{{\qubit{\psi}}}}
\newcommand{\qphi}{\ensuremath{{\qubit{\varphi}}}}
\newcommand{\enum}[2]{\ensuremath{{#1_1#2#1_2#2\ldots#2#1_n}}}

\newcommand{\Ent}{\ensuremath{{E_{Hmin}}}}

\newcommand{\shannon}{\ensuremath{H_{sh}}}
\newcommand{\shrho}[1]{\shannon(Diag({#1}))}
\newcommand{\vn}{\ensuremath{H_{vN}}}
\newcommand{\hmes}{\ensuremath{H_{meas}}}

\begin{document}

\title{Entanglement measure for multipartite pure states and its numerical calculation}
\author{A. Yu.~Chernyavskiy}
\email{andrey.chernyavskiy@gmail.com}
\affiliation{Institute of
Physics and Technology, Russian Academy of Sciences} 
\begin{abstract}

The quantification and classification of quantum entanglement is a
very important and still open question of quantum information
theory. In this paper, we describe an entanglement measure for
multipartite pure states (the minimum of Shannon's entropy of
orthogonal measurements). This measure is additive, monotone under
LOCC, and coincides with the reduced von Neumann entropy on
bipartite states. A method for numerical calculation of this measure
by genetic algorithms is also presented. Moreover, the minimization
of entropy technique is extended to fermionic states.

\end{abstract}
\maketitle

\section{Introduction}
Multipartite quantum entanglement plays significant role in whole
quantum science. For example, one can easily obtain that we can't
achieve algorithmic speedup in quantum computer without quantum
entanglement. Also entanglement is needed for other quantum
protocols: quantum teleportation, quantum error correction, etc.
What is more, the theory of quantum entanglement may help us to
understand multiparticle quantum physics deeply.

It's well known that the entanglement of bipartite pure states is
fully described. Schmidt coefficients uniquely define a local
unitary orbit of a bipartite quantum state $\qpsi$ and unambiguously
determine the class of states that can be obtained from $\qpsi$ by
LOCC \cite{nielsen2000qca}.

Schmidt decomposition
$$\sum\limits_{i,j} {{c_{ij}}} |i{\rangle _A} \otimes |j{\rangle _B} = \sum\limits_i {\sqrt {{\lambda _i}} } |i{\rangle _A} \otimes |i{\rangle _B}$$
is an equivalent of matrix SVD decomposition
$$\begin{array}{c}
 \sum\limits_{i,j} {{c_{ij}}} |i{\rangle _A}{\langle j{|_B} = \sum\limits_i {\sqrt {{\lambda _i}} } |i\rangle _A}\langle i{|_B}, \\
 A = US{V^*}. \\
 \end{array}
$$

It's easy to see that there is no Schmidt decomposition even for
three qubits. The $W$ state $$\frac{1}{{\sqrt 3 }}(|001\rangle  +
|010\rangle  + |100\rangle )$$ cannot be represented as $${\lambda
_0}|\widetilde0\rangle  \otimes |\widetilde0\rangle  \otimes
|\widetilde0\rangle  + {\lambda _1}|\widetilde1\rangle  \otimes
|\widetilde1\rangle  \otimes |\widetilde1\rangle.$$ A stronger
counterexample exists for SVD decomposition: there is no lower-rank
orthogonal approximation even for third-order tensors
\cite{kolda2001cpe}.

While a lot of applications use SVD: image compression (for example,
\cite{waldemar1997hks}), LSI (Latent Semantic Indexing, for example,
\cite{deerwester1990ils}), statistics (for example,
\cite{hammarling1985svd}), etc., the appropriate expanding of SVD
decomposition (Schmidt decomposition) to higher order tensors (to
more than two subsystems) is an important task not only for quantum
information theory but for fundamental (tensor algebra) and applied
(statistics, machine learning, etc.) math. Some higher order
analogues of SVD may be found (for example,
\cite{lathauwer1995msv}), but non of them can be directly applied to
the quantum pure state case.

Whereas there is no full theory of multipartite quantum
entanglement, some important results are known.

A full classification of pure states of three qubits in terms of
SLOCC (stochastic LOCC) is given in \cite{dur2000tqc}. A description
of quantum entanglement with nilpotent polynomials is constructed in
\cite{mandilara2005dqe}. Some entanglement measures of pure and
mixed states can be found, for example, in
\cite{plenio2005lnf,vidal2002qpt,horodecki2007qe,wong2000pme,
brennen2003ome,vedral1997ema}. Theory of quantum entanglement
measures is studied in, for example,
\cite{plenio2005iem,bouwmeester2000pqi,bruss2002ce,vidal2000em,amico2007emb}.

The paper is organized as follows: In the first section of this
paper we present the minimum of Shannon's entropy of orthogonal
measurements and prove that this function possesses all necessary
pure state entanglement measure properties. Some other features of
this measure and its values for some states are also presented. Next
we describe a method to compute this measure using genetic
algorithms. An extension of minimal measurements entropy to
fermionic states is described in the last part of the paper.

\subsection{Notation}
$\shannon$ is the Shannon entropy,

$\vn$ is the von Neumann entropy.
\section{Shannon's entropy of orthogonal measurements minimum}

Consider a pure $n$-qudit qauntum state
\begin{equation}
\label{quditState}
\qpsi=\sum\limits_{\enum{i}{,}=1}^{d}a_{\enum{i}{,}}\qubit{\enum{i}{}}.
\end{equation}

The measurement of this state in the computational basis gives
states $\qubit{\enum{i}{}}$ with probabilities
$|a_{\enum{i}{,}}|^2$. So we can regard this measurement process as
a signal generator. We shall say that Shannon's entropy of this
generator is called \emph{measurements entropy} of the state
$\qpsi$:
$$\hmes(\qpsi)=\shannon(Diag(\qrho{\psi}))=\sum\limits_{\enum{i}{,}=1}^{d}|a_{\enum{i}{,}}|^2.$$
But this characteristic of a quantum state is not invariant under
local changes of the measurement basis. Shannon's source coding
theorem shows that, in the limit, the average length of the
\emph{shortest} possible representation to encode the messages in a
given alphabet is their entropy divided by the logarithm of the
number of symbols in the target alphabet. That's why the natural way
to construct an invariant is the minimization of $\hmes(\qpsi)$ over
all possible local changes of the measurements basis:
\begin{equation}
\label{edef}
\Ent(\qpsi)=\min\limits_{\enum{U}{,}}
\hmes(\enum{U}{\otimes}\qpsi);
\end{equation} and $\Ent$ is a measure of entanglement of the pure
state $\qpsi$.

\begin{rem}
\label{remQudit} The definition of $\Ent$ for subsystems with
different dimensions is the same to (\ref{edef}). As non of the
following reasoning bears on the equality of subsystem dimensions,
here and further the common $d$ for different $i_j$ in
(\ref{quditState}) is used only for notation simplicity.
\end{rem}

\subsection{Necessary properties of entanglement measure}

Entanglement measure for pure quantum states must satisfy the
following intuitive conditions:

(i) must be equal to zero for fully unentangled states

(ii) must be invariant under local unitary operations

(iii) must be invariant under attachment and detachment of
unentangled ancilla

(iv) must not increase under LOCC.

Let's show that $\Ent$ satisfies (i)-(iv). First of all we consider
another one important property:

\begin{prop}(Additivity of \Ent)

Let $\qubit{\psi_1}$ and $\qubit{\psi_2}$ be multipartite (may be
entangled) pure states, then
$$\Ent(\qubit{\psi_1} \otimes \qubit{\psi_2}) = \Ent(\qubit{\psi_1})
+ \Ent(\qubit{\psi_2}).$$
\end{prop}

The proof is straightforward from the definition of $\Ent$ and
additivity of Shannon's entropy.

The properties (i)-(ii) are trivial and (iii) is the consequence of
$\Ent$ additivity. To prove (iv), since (ii) and (iii) are correct,
we may only prove that $\Ent$ doesn't increase under orthogonal
measurements in the mean. To show monotonicity under orthogonal
measurements in the mean we need some lemmas.

\begin{lem}
\label{LemmaRho1} Let $\rho_{AB}$ be a density matrix of a pure
bipartite state. Then
$$\shrho{\rho_{A}}\leq \shrho{\rho_{AB}} \leq \shrho{\rho_{A}} + \shrho{\rho_{B}},$$
where $\rho_{A}$ and $\rho_{B}$ are reduced density matrices.
\end{lem}
\begin{proof}
The left inequality follows from generalized grouping of Shannon's
entropy:
$$\shannon(p_1,...,p_{\sigma_1},p_{\sigma_1+1},...,p_{\sigma_2},...,p_{\sigma_{n-1}+1},...,p_{\sigma_n})$$
$$=\shannon(p_1+...+p_{\sigma_1},p_{\sigma_1+1}+...+p_{\sigma_2},...,p_{\sigma_{n-1}+1}+...+p_{\sigma_n})$$
$$ +\,\sum_{i=1}^n{(p_{\sigma_{i-1}+1}+...+p_{\sigma_i})\shannon
\Big(p_{\sigma_{i-1}+1}\Big/\sum_{j=\sigma_{i-1}+1}^{\sigma_i}{p_j},...,p_{\sigma_i}\Big/\sum_{j=\sigma_{i-1}+1}^{\sigma_i}{p_j}\Big)}
.$$ The right part follows the sub-additivity of entropy:
$$\shannon(v_{11},v_{12},...,v_{1m},v_{21},...,v_{2m},...,v_{n1},...,v_{nm}\Big)$$
\begin{displaymath}
\leq
\shannon\Big(\sum_{i=1}^n{v_{i1}},\sum_{i=1}^n{v_{i2}},...,\sum_{i=1}^n{v_{im}}\Big)
+\shannon\Big(\sum_{j=1}^m{v_{1j}},\sum_{j=1}^m{v_{2j}},...,\sum_{j=1}^m{v_{nj}}\Big).
\end{displaymath}
\end{proof}

\begin{lem}
\label{LemmaMonotone1} Let
$$\qpsi=\sum\limits_{\enum{i}{,}=1}^{d}a_{\enum{i}{,}}\qubit{\enum{i}{}}$$
be a pure $n$-qudit state. The states $\qubit{j}\otimes
\qubit{\psi_j}$ are results of the measurement of $\qpsi$ in
computational basis with probabilities $p_j$. Then
\begin{equation}
\label{MainGeq} \hmes(U_1 \qpsi) \geq \sum\limits_{j=1}^{d}p_j
\hmes(\qubit{j}\otimes \qubit{\psi_j}),
\end{equation}
 where $U_1$ is an arbitrary unitary transformation of the first qudit.
\end{lem}
\begin{proof}
If $\qubit{j}\otimes \qubit{\psi_j}$ are results of $\qpsi$
measurement in computational basis with probabilities $p_j$, then
$\qpsi$ can be represented in the view
$$\qpsi = \sum\limits_{j=1}^{d}c_j \qubit{j}\otimes
\qubit{\psi_j},$$ where
$$|c_j|^2=p_j.$$

Let $\psi_j^k,k=\overline{1,d^{n-1}}$ be the amplitudes of
$\qubit{\psi_j},$ then
$$\hmes(\qpsi)=\hmes(\sum\limits_{j=1}^{d}c_j \qubit{j}\otimes
\qubit{\psi_j}) = $$
$$=\shannon(|c_1|^2 |\psi_1^1|^2,\ldots,|c_1|^2
|\psi_1^{d^{n-1}}|^2,\cdots,|c_d|^2 |\psi_d^1|^2,\ldots,|c_d|^2
|\psi_d^{d^{n-1}}|^2)$$

\begin{center}
$=\{$ by the strong additivity of Shannon's entropy $\}$
\end{center}

$$=\sum\limits_{j=1}^d {p_j \shannon(|\psi_j^1|^2 , |\psi_j^2|^2, \ldots, |\psi_j^{d^{n-1}}|^2)} + \shannon(p_1, p_2, \ldots, p_d)$$
$$=\sum\limits_{j=1}^{d}p_j
\hmes(\qubit{j}\otimes \qubit{\psi_j})+\shannon(\mathbf{p}).$$

Using this we can rewrite (\ref{MainGeq}) in the equivalent form
\begin{equation}
\label{MainGeq2}  \hmes(\qpsi) \leq \hmes(U_1 \qpsi) +
\shannon(\mathbf{p}).
\end{equation}

Let $\rho=\qrho{\psi}$ be the density matrix of the state \qpsi,
$\rho_1 = Tr_{2,3,\ldots,d}(\rho)$, $\rho_2 = Tr_1 (\rho)$. Then by
Lemma \ref{LemmaRho1} for $\rho$ we get
$$\shrho{\rho} \leq \shrho{\rho_2}+\shannon(Diag(\rho_1)),$$
or equivalently
\begin{equation}
\label{eq1} \hmes(\qpsi) \leq \shrho{\rho_2}+\shannon(\mathbf{p}).
\end{equation}

Let $\rho_{U_1}=U_1\qrho{\psi}U_1^*$ be the density matrix of the
state $U_1\qpsi$, $\rho_2^{U_1} = Tr_1 (\rho_{U_1})$. Since $U_1$
affects the first qudit only, $\rho_2^{U_1} = \rho_2.$

By Lemma \ref{LemmaRho1}
\begin{equation}
\label{eq2} \shrho{\rho_2} = \shrho{\rho_2^{U_1}} \leq
\shrho{\rho_{U_1}}.
\end{equation}

By (\ref{eq1}) and (\ref{eq2}) it follows that (\ref{MainGeq2}) is
correct, consequently (\ref{MainGeq}) is correct too.
\end{proof}

\begin{thm}\label{mainThm}(the monotonicity of  $\Ent$ under orthogonal measurements)

Let $\qpsi$ be an $n$-qudit state. The states $\qubit{\psi_j}$ are
results of some orthogonal measurement with probabilities $p_j$.
Then
$$\Ent(\qpsi)\geq \sum\limits_{j} p_j \Ent(\qubit{\psi_j}).$$
\end{thm}

\begin{proof}
We may assume, without loss of generality, that we measure first
qudit in the basis $\{\psi^1_i\}$ (a measurement of more than one
qudit can be replaced by sequential one-qudit measurements). Then
$\qubit{\psi_j}=\qubit{\psi^1_j}\otimes\qubit{\psi^2_j}$. Let
$\qubit{\psi_{min}}$ has minimal measurements entropy over local
unitary orbit of $\qpsi$, and $\qubit{\psi_{min}} =
\enum{U}{\otimes}\qubit{\psi}$. Then
$$\sum\limits_{j} p_j \Ent(\qubit{\psi_j}) =
\sum\limits_{j} p_j \Ent(\qubit{\psi^1_j}\otimes \qubit{\psi^2_j})
\leq \sum\limits_{j} p_j \hmes(\qubit{j}\otimes
(U_2\otimes\ldots\otimes U_n)\qubit{\psi_j})$$

$\leq\{$by Lemma \ref{LemmaMonotone1}$\} \leq \hmes(\psi_{min}) =
\Ent(\qpsi).$
\end{proof}
Thus, from \ref{mainThm} we get (iv).

\subsection{Other properties}
\begin{lem}(Klein's lemma)

Let $\rho$ be a density matrix, then
$$\shannon(Diag(\rho)) \geq \vn(\rho).$$
\end{lem}

\begin{thm}
\label{thmVn} $\Ent$ coincides with the reduced von Neumann entropy
for bipartite states. I.e., let $\qpsi$ be a pure bipartite state;
$\rho=\qrho{\psi}$ is its density matrix, $\rho_A=Tr_B(\rho)$ and $
\rho_B=Tr_A(\rho)$ are reduced density matrices of subsystems. Then
$\Ent(\qpsi)=\vn(\rho_A)=\vn(\rho_B)$.
\end{thm}
\begin{proof}
Consider the Schmidt decomposition of $\qpsi$:
$$\qpsi = \sum\limits_i \sqrt{\lambda_i}\qubit{i_A}\otimes \qubit{i_B}.$$
Then

$$\Ent(\qpsi)=\min\limits_{U_A,U_B}{\hmes(U_A \otimes U_B \qpsi)}\leq \shannon(\lambda_i) = \vn(\rho_A) = \vn(\rho_B).$$

From Lemma \ref{LemmaRho1} we get
$$\hmes(\qpsi) = \shrho{\rho} \geq \shannon(Diag(\rho_A)).$$
Further, by Klein's lemma we have
$$\shannon(Diag(\rho_A)) \geq \vn(\rho_A) =\vn(\rho_B).$$

Thus $\Ent(\qpsi)=\vn(\rho_A)=\vn(\rho_B).$
\end{proof}

\begin{prop}($\Ent$ of generalized GHZ states)
Consider a generalized GHZ-state:
$$\qubit{GHZ}=\sum\limits_{i=1}^d a_i{\enum{\qubit{i}}{\otimes}},$$
where $a_i \in \Real, \sum\limits_{i=1}^d |a_i|^2=1.$

Then
$$\hmes(\qubit{GHZ})=\min\limits_{\enum{U}{,}} \hmes(\enum{U}{\otimes}\qubit{GHZ}).$$
\end{prop}
\begin{proof}
We can consider the space of GHZ state as bipartite space of the
first qudit and the remaining qudits. Then
$$\qubit{GHZ}=\sum\limits_{i=1}^d a_i\qubit{i}_1\otimes \qubit{\widetilde{i}}.$$
From Theorem \ref{thmVn} and Remark \ref{remQudit} we have that
Schmidt decomposition has minimal measurements entropy. From this we
get
$$\hmes(\qubit{GHZ})=\hmes(\sum\limits_{i=1}^d a_i\qubit{i}_1\otimes \qubit{\widetilde{i}})\leq \hmes(U_1 \otimes \widetilde{U} \qubit{GHZ}),$$
where $U_1$ is a unitary transformation of the first qudit and
$\widetilde{U}$ is an arbitrary unitary transformation of
$2,3,\cdots,n$ qudits. We can take $U_2\otimes U_3 \otimes \cdots
\otimes U_n$ as $\widetilde{U}$, then
$$\hmes(\qubit{GHZ})\leq
\hmes(\enum{U}{\otimes}\qubit{GHZ}),$$
this finishes the proof.
\end{proof}
So we have $\Ent(\qubit{GHZ})=-\sum\limits_{i=1}^d |a_i|^2\ln
|a_i|^2.$

\subsection{Numerical properties}
The following properties were obtained numerically by genetic
algorithms.
\begin{numprop}($\Ent$ of generalized W states)

Consider a generalized W state
$$\qubit{W}=a_1\qubit{0\ldots01}+a_2\qubit{0\ldots10}+a_n\qubit{1\ldots00},$$
where $\sum\limits_{i=1}^{n}\abs{a_i}^2=1$. Then
$$\Ent(W) = \hmes(W) =
\sum\limits_1^n{{\abs{a_{i}}^2\ln{\abs{a_{i}}^2}}}.$$
\end{numprop}

\begin{numprop}
Let
$$\qpsi=\sum\limits_{\enum{i}{,}=1}^{d}a_{\enum{i}{,}}\qubit{\enum{i}{}},$$
$$\qphi=\sum\limits_{\enum{i}{,}=1}^{d}b_{\enum{i}{,}}\qubit{\enum{i}{}},$$
and $\hmes(\qpsi)=\hmes(\qphi)=\Ent(\qpsi)=\Ent(\qphi)$ (in other
words $\qpsi$ and $\qphi$ have equal $\Ent$ and they are in minimal
entropy representation), then
$$|a_{\enum{i}{,}}|^2 = |b_{\enum{i}{,}}|^2$$
to within local permutations of basis vectors.

I.e. modulus squares of a minimal entropy representation of the
local unitary orbit are unique.
\end{numprop}

\subsection{$\mathbf{\Ent}$ is substantially multipartite}
Numerical computations using genetic algorithms show that two
equivalent under bipartite entanglement states of three qubits
$\qphi$ and $\qpsi$ may have different $\Ent.$ This means that
$\Ent$ is substantially multipartite. Equivalence under bipartite
entanglement means

$$\begin{array}{c}
\exists U^1_1,U^1_{23},U^2_2,U^2_{13},U^3_3,U^3_{12}:\\
\qphi=U^1_1\otimes U^1_{23}\qpsi,\\
\qphi=U^2_2\otimes U^2_{13}\qpsi,\\
\qphi=U^3_3\otimes U^3_{12}\qpsi,
\end{array}
$$
where $U^k_i$ are unitary evolutions of $i$-th qubit, and $U^k_{ij}$
are unitary evolutions (may be entangled) of $i$-th and $j$-th
qubit.

From $\Ent(\qpsi)\neq\Ent(\qphi)$ we have that $\qphi$ and $\qpsi$
are not equivalent under local unitary transformations, i.e
$$\nexists U_1,U_2,U_3:\qphi=U_1\otimes U_2\otimes U_3 \qpsi.$$

This property of $\Ent$ gives an answer to the problem of
equivalence of bipartite and multipartite entanglement proposed in
\cite{burkov2006aaq}: the bipartite and multipartite entanglements
are not equivalent. This also means that entanglement measures based
only on Schmidt coefficients of different decompositions of a state
are not good for quantifying multipartite entanglement.

\section{Calculation of $\mathbf{\Ent}$ using genetic algorithms}
\subsection{Formalization of the optimization problem}
\label{formalization}Consider the $n$-qudit state
$$\qpsi=\sum\limits_{\enum{i}{,}=1}^{d}a_{\enum{i}{,}}\qubit{\enum{i}{}}.$$

Then
$$\Ent(\qpsi)=\min\limits_{\enum{U}{,}}\hmes(\enum{U}{\otimes}\qpsi),$$
where $U_i$ is a $d$-dimension unitary operator on $i$-th qudit,
$$\hmes(\qpsi)=\sum\limits_{\enum{i}{,}=1}^{d}\abs{a_{\enum{i}{,}}}^2\ln{\abs{a_{\enum{i}{,}}}^2}.$$

To solve this optimization problem we need to parameterize unitary
matrices $U_i$. When $d=2$ this parametrization is well known:
\begin{equation}
U_i(\beta_i,\delta_i,\gamma_i)=\left(
              \begin{array}{cc}
                e^{i(-\beta_i-\delta_i)}\cos\gamma_i & -e^{i(-\beta_i+\delta_i)}\sin\gamma_i \\
                e^{i(\beta_i-\delta_i)}\sin\gamma_i & e^{i(\beta_i+\delta_i)}\cos\gamma_i \\
              \end{array}
            \right),
\end{equation}
where $\beta_i,\delta_i,\gamma_i$ are real numbers.

Parametrization of unitary matrices for $d>2$ is a still open and
interesting question. The best known parametrization is proposed in
\cite{tilma2004gea,tilma2002gea}, but it is very slow, so we have
chosen another one. Let's parameterize hermitian matrix $H$ by $d^2$
real numbers (the diagonal has $d$ real numbers, and $d(d-1)/2$
complex numbers above the diagonal need $d(d-1)$ real numbers to be
parameterized). Then we take $U=e^{iH}$  as a unitary matrix.

In such a way the calculation of $\Ent$ is an optimization task of
$3n$ real parameters for qubits and $nd^2$ real parameters for
qudits ($d>2$). To calculate $\Ent(\qpsi)$ we need to minimize the
function
$$f^\psi(x_1,\ldots,x_{n\cdot k})=\hmes(U(x_1,\ldots,x_k)\otimes U(x_{k+1},\ldots,x_{2k}) \otimes \ldots \otimes U(x_{(n-1)k+1},\ldots,x_{n\cdot k})\qpsi),$$
where $U(x_1,\ldots,x_k)$ is a parametrization of $d\times d$
unitary matrix by $k$ real parameters.

In the general case, functions $f^\psi(x_1,\ldots,x_{n\cdot k})$ may
be multimodal (have many local minimums), because of this fact we
can't use gradient-based optimization methods. Thus genetic
algorithm has been chosen.

\subsection{Genetic algorithm (GA)}
There are a lot of publications about GA, but for a brief overview
and references we recommend Wikipedia \cite{wiki:ga}.

GA has been already used for quantum entanglement calculation (to
calculate relative entropy of entanglement of mixed bipartite
states) in \cite{ramos2002cqe}.

Now we describe GA that was used for $\Ent$ calculations.

Every parameter $x$ of $f^\psi$ is being encoded by $n_{gen}$ real
value genes $g_{j},j=\overline{1,n_{gen}}$ by the rule
$x=\sum\limits_{j=1}^{n_{gen}}10^{1-j}g_{j}$. So, if we have $k$
parameters, then a chromosome is a vector $\{g_i\}$ of length
$n_{gen}\cdot k$. (A full tuple of parameters of $f^\psi$ is encoded
by a chromosome.)

A mutation of a chromosome $\{g_i\}$ is determined by the following
probabilities:
$$P(g_i^m=g_i)=(1-p_{mut}),\qquad
P(g_i^m=g_i+\xi)=p_{mut},$$ where $\{g_i^m\}$ is a chromosome after
mutation, $p_{mut}$ is a mutation probability, $\xi$ is a random
variable that has uniform distribution on $[-m_{mut},m_{mut}]$.

Crossover of chromosomes $\{g_i^1\}$ and $\{g_i^2\}$ is a chromosome
$\{g_i^r\}$, where
$$P(g_{i}^r=g_{i}^1)=P(g_{i}^r=g_{i}^2)=\frac{1}{2}.$$

A fitness function is \quad $-f^\psi.$

\begin{samepage}
\textit{The algorithm.}
\begin{enumerate}
\item Initialize the first population of $n_{population}$ random
chromosomes with uniformly distributed on $[-m_{init},m_{init}]$
genes.

\item Wait an epoch (this step will be described further). After the epoch we have a new population.

\item If one of termination conditions is satisfied, algorithm stops and the
fitness of the best chromosome of the last population is taken as a
result, else we repeat from Step 2.
\end{enumerate}
\end{samepage}
\bigskip

\begin{samepage}
\textit{Termination conditions:}

We can fix the maximal number of epochs($n_{epochs}$), the precision
$\eps$, and the maximal number of nonchanging epoches $n_{term}$.
Then termination conditions will be:

--We reach $n_{epochs}$ epoche.

--Examine the best chromosome from each of the last $n_{term}$
epoches. If the fitness function values of these chromosomes differ
from each other by less than $\eps$, then this termination condition
is true.
\end{samepage}
\bigskip

\begin{samepage}
\textit{The epoch.}
\begin{enumerate}

\item Two random pairs of chromosomes are chosen from the best $n_{population}-n_{bad}$ chromosomes of the
population, $n_{bad}$ is a number of the weakest chromosomes of the
population, which cannot be used for reproduction. From each pair we
take a chromosome with the best fitness. After that we put a
crossover result of the two selected chromosomes to the new
population.

\item Repeat Step 1 while the new population is lesser than
$n_{population}$.

\item Mutate every chromosome of the new population.
\end{enumerate}
\end{samepage}

While $f^\psi$ for $\qpsi$ from local orbits of "easy" states like
GHZ or W has no "difficult" local minimums, the situation is reverse
for states with random amplitudes. For example, one $7$-qubit state
has a local minimum $~4.0220$, while another one (probably global)
minimum is $~3.968$. To solve this problem we use GA with "islands":
we form $n_{islands}$ islands with rather small equal populations.
Epoches at these islands are independent and only infrequent
migrations are allowed. The increase of islands count is an
equivalent of the the whole GA repeating, so the probability of an
error decreases exponentially with the number of islands. Increase
of the probability of migration speedups the convergence, but it
also increases the probability of local minimums results.

\bigskip

\emph{\textbf{Remark}}\emph{(about software implementation)}

One of the advantages of GA is its parallelism, so the software
realization of $\Ent$ calculation was multithreading. Moreover, the
implementation of $\Ent$ calculation for qubits was realized using
nVidia CUDA \cite{cuda} technology. Using cheap personal GPU,
\rm{7x} speedup against the best 4-core CPU implementation has been
achieved. Software complex can easily calculate $\Ent$ for up to 17
qubits inclusively.

\section{$\mathbf{\Ent}$ for fermionic states}

Consider a pure state of $n$ fermions in $p$-dimensional Hilbert
space with basis $\enum{f}{,}$:
$$\qubit{f}=\sum\limits_{\enum{i}{,}=1,i_1<i_2< \ldots <i_n}^{p}\lambda_{\enum{i}{}} \qubit{\enum{i}{}},$$
where
$$\sum\limits_{\enum{i}{,}=1,i_1<i_2< \ldots <i_n}^{p}|\lambda_{\enum{i}{}}|^2=1,$$
$$\qubit{\enum{i}{}}=\frac{1}{\sqrt{n!}}\left|
\begin{array}{*{20}{c}}
    f_{i_1}(r_1) &  \ldots  & f_{i_1}(r_n)  \\
    \vdots  &  \ddots  &  \vdots   \\
    f_{i_n}(r_1) &  \cdots  & f_{i_1}(r_n)  \\
\end{array}
\right|$$ are Slater determinants.

This states correspond to normalized elements of the exterior
algebra $\Lambda^n \mathbb{C}^k$. Let's take states that correspond
to separable elements $\enum{x}{\bigwedge}$ of $\Lambda^n
\mathbb{C}^k$ as unentangled, i.e. the states that are represented
as a single Slater determinant in \emph{some basis} of one-particle
$\mathbb{C}^k$ space are unentangled. Thus, we can take unitary
changes of basis as unentangling transforms.

This is more mathematical than physical approach. For example, if we
have two electrons in two quantum dots their one-particle space must
at the least include spin and position coordinates. In the simplest
case the whole space will be $H_{spin}\otimes H_{position}$ and its
basis will be $\qubit{\uparrow}\otimes\qubit{1}$,
$\qubit{\uparrow}\otimes\qubit{2}$,
$\qubit{\downarrow}\otimes\qubit{1}$,
$\qubit{\downarrow}\otimes\qubit{2}$. As we can see, unitary
transformations of this basis are ``entangled'' in an intuitive
physical way. So, the construction of a physical formalism of
unentangled local transformations of identical particles is an
important direction for future research.

But mathematical approach is nevertheless very important and is used
in many papers devoted to the entanglement of indistinguishable
particles. For example, the Slater decomposition (an analogue of the
Schmidt decomposition) of two fermions is constructed in
\cite{schliemann2000qct}, and we will discuss it further. Another
interesting result is a classification of $\Lambda^3 \mathbb{C}^6$
fermionic states in terms of SLOCC \cite{levay-2008-78}.

Now let's define $\Ent$ for fermionic states.

The change of basis is determined by the unitary matrix $U$ and in
the new basis the state becomes
$$U \circ \qubit{f}=\sum\limits_{\enum{i}{,}=1,i_1<i_2< \ldots <i_n}^{p}\lambda^{'}_{\enum{i}{}} \qubit{\enum{i}{}},$$
$$\lambda^{'}_{\enum{j}{,}} = \sum\limits_{\enum{i}{,}=1, i_1<i_2< \ldots <i_p}^{p} \lambda_{\enum{j}{,}}M^{{\enum{j}{,}}}_{\enum{i}{,}},$$
where $M^{{\enum{j}{,}}}_{\enum{i}{,}}$ is a determinant of the
matrix that is constructed by the crossing of \enum{j}{,} columns
and \enum{i}{,} rows of $U$.

Measurements entropy of $\qubit{f}$ is
$$\hmes(\qubit{f})=\sum\limits_{\enum{i}{,}=1,i_1<i_2< \ldots <i_n}^{p}|\lambda_{\enum{i}{}}|^2 \log |\lambda_{\enum{i}{}}|^2.$$
And
$$\Ent(\qubit{f})=\min\limits_{U}\hmes(U\circ\qubit{f}).$$

The parametrization of $U$ is described in Section
\ref{formalization}. To get $U\circ\qubit{f}$ amplitudes we need to
calculate all minors of $U$, this can be done recursively using
determinant expansion by minors.

\subsection{Slater decomposition}
Now consider a state of $2$ fermions in $2p$-dimensional space:
$$\qubit{f_2}=\sum\limits_{i_1,i_2=1, i_1<i_2}^{2p}\lambda_{i_1i_2}\qubit{i_1i_2}.$$
Then this state can be represented as a linear combination of only
$p$ Slater determinants using a unitary change of basis
\cite{schliemann2000qct}:
$$U_{slater}\circ\qubit{f_2} = \sum\limits_{i=1}^{p}z_{i}\qubit{2i,2i+1}.$$
This representation is called \emph{Slater decomposition}.

This fact follows from the existence of the symplectic basis for
antisymmetric matrices \cite{vinberg2003ca}, full proof may be found
in \cite{schliemann2000qct}.

An important \emph{numerical} result has been achieved for Slater
decomposition and $\Ent$
\begin{num}Let $\qubit{f_2}$ be the state of $2$ fermions with $2p$-dimensional
one-particle space. Then the following two conditions are
equivalent:

1. The state $\qubit{f_2}$ is in the Slater decomposition
representation:
$$\qubit{f_2} = \sum\limits_{i=1}^{p}z_{i}\qubit{2i,2i+1}.$$

2. $\Ent(\qubit{f_2})=\hmes(\qubit{f_2}).$

(I.e. Slater decomposition, like Schmidt decomposition, has minimal
measurements entropy, and the process of finding $\Ent$ for a state
of two fermions is equal to finding its Slater decomposition.)

\end{num}

\section{Conclusion}
Entanglement measure $\Ent$ for multipartite pure states  was
presented. Also we proved that this measure is additive, satisfies
all necessary entanglement measure conditions,  and coincides with
the reduced von Neumann entropy for bipartite states. The method of
numerical calculation of $\Ent$ by genetic algorithm was presented
and tested on up to 17 qubits inclusively. Moreover, $\Ent$ was
generalized to fermionic states, and this generalization corresponds
to Slater decomposition for two-fermions states.

\begin{acknowledgments}
The author is grateful to professor Yu.I.~Ozhigov for problem
formulation and constant attention to this work, to professor
Yu.I.~Bogdanov for helpful discussions, and to professor A.S.~Holevo
for very useful advice to pay attention to the book ``Matrix
analysis'' by R.~Bhatia \cite{bhatia1997ma}.

\bigskip

This work is partially supported by RFBR grant 09-01-00347-a.
\end{acknowledgments}
\bibliography{mybib}{}
\bibliographystyle{unsrt}

\end{document}